\makeatletter
\def\input@path{{styles/}}
\makeatother

\documentclass[12pt]{article}
\def\UseBibLatex{1}

\providecommand{\BibLatexMode}[1]{}
\providecommand{\BibTexMode}[1]{}

\ifx\UseBibLatex\undefined%
  \renewcommand{\BibLatexMode}[1]{}
  \renewcommand{\BibTexMode}[1]{#1}
\else
  \renewcommand{\BibLatexMode}[1]{#1}
  \renewcommand{\BibTexMode}[1]{}
\fi

\BibLatexMode{%
   \usepackage[bibencoding=utf8,style=alphabetic,backend=biber]{biblatex}%
   \usepackage{sariel_biblatex}%
}

\usepackage{amsmath}%
\usepackage{amssymb}%
\usepackage[table]{xcolor}%
\usepackage{mathcalb}

\usepackage[amsmath,thmmarks]{ntheorem}%
\usepackage{titlesec}%
\usepackage{xcolor}%
\usepackage{mleftright}%
\usepackage{xspace}%
\usepackage{graphicx}
\usepackage{microtype}
\usepackage[inline]{enumitem}

\usepackage{hyperref}%
\usepackage[ocgcolorlinks]{ocgx2}

\hypersetup{%
      unicode,
      breaklinks,%
      colorlinks=true,%
      urlcolor=[rgb]{0.25,0.0,0.0},%
      linkcolor=[rgb]{0.5,0.0,0.0},%
      citecolor=[rgb]{0,0.2,0.445},%
      filecolor=[rgb]{0,0,0.4},
      anchorcolor=[rgb]={0.0,0.1,0.2}%
}

\titlelabel{\thetitle. }%

\theoremseparator{.}%

\theoremstyle{plain}%
\newtheorem{theorem}{Theorem}[section]

\newtheorem{lemma}[theorem]{Lemma}

\theoremstyle{plain}%
\theoremheaderfont{\sf} \theorembodyfont{\upshape}%
\newtheorem*{remark:unnumbered}[theorem]{Remark}%
\newtheorem{remark}[theorem]{Remark}%

\theoremheaderfont{\em}%
\theorembodyfont{\upshape}%
\theoremstyle{nonumberplain}%
\theoremseparator{}%
\theoremsymbol{\myqedsymbol}%
\newtheorem{proof}{Proof:}%

\providecommand{\emphind}[1]{}%
\renewcommand{\emphind}[1]{\emph{#1}\index{#1}}

\definecolor{blue25emph}{rgb}{0, 0, 11}

\providecommand{\emphic}[2]{}
\renewcommand{\emphic}[2]{\textcolor{blue25emph}{%
      \textbf{\emph{#1}}}\index{#2}}

\providecommand{\emphi}[1]{}%
\renewcommand{\emphi}[1]{\emphic{#1}{#1}}

\definecolor{almostblack}{rgb}{0, 0, 0.3}

\providecommand{\emphw}[1]{}%
\renewcommand{\emphw}[1]{{\textcolor{almostblack}{\emph{#1}}}}%

\providecommand{\emphOnly}[1]{}%
\renewcommand{\emphOnly}[1]{\emph{\textcolor{blue25emph}{\textbf{#1}}}}

\newcommand{\myqedsymbol}{\rule{2mm}{2mm}}
\newcommand{\SarielThanks}[1]{%
   \thanks{%
      School of Computing and Data Science; %
      University of Illinois; %
      201 N. Goodwin Avenue; %
      Urbana, IL, 61801, USA; %
      \href{mailto:spam@illinois.edu}{sariel@illinois.edu}; %
      \url{http://sarielhp.org/}.%
   #1%
   }%
}

\newcommand{\HLink}[2]{\hyperref[#2]{#1~\ref*{#2}}}
\newcommand{\HLinkY}[2]{\hyperref[#2]{#1}}
\newcommand{\HLinkSuffix}[3]{\hyperref[#2]{#1\ref*{#2}{#3}}}

\newcommand{\figlab}[1]{\label{fig:#1}}
\newcommand{\figref}[1]{\HLink{Figure}{fig:#1}}

\newcommand{\thmlab}[1]{{\label{theo:#1}}}

\newcommand{\thmrefY}[2]{\HLinkY{#2}{theo:#1}}

\newcommand{\lemlab}[1]{\label{lemma:#1}}
\newcommand{\lemref}[1]{\HLink{Lemma}{lemma:#1}}%

\providecommand{\eqlab}[1]{}%
\renewcommand{\eqlab}[1]{\label{equation:#1}}
\newcommand{\Eqref}[1]{\HLinkSuffix{Eq.~(}{equation:#1}{)}}

\providecommand{\remove}[1]{}%

\newcommand{\pth}[1]{\mleft(#1\mright)}%

\newcommand{\ProbC}{{\mathbb{P}}}
\newcommand{\ExC}{{\mathbb{E}}}

\newcommand{\ExCond}[2]{\ExC\!\left[%
       #1 \;\middle\vert\; #2 \right]}

\newcommand{\Prob}[1]{\ProbC\mleft[ #1 \mright]}
\newcommand{\Ex}[1]{\ExC\mleft[ #1 \mright]}

\newlist{compactenumA}{enumerate}{5}%
\setlist[compactenumA]{topsep=0pt,itemsep=-1ex,partopsep=1ex,parsep=1ex,%
   label=(\Alph*)}%

\newlist{compactenuma}{enumerate}{5}%
\setlist[compactenuma]{topsep=0pt,itemsep=-1ex,partopsep=1ex,parsep=1ex,%
   label=(\alph*)}%

\newlist{compactenumI}{enumerate}{5}%
\setlist[compactenumI]{topsep=0pt,itemsep=-1ex,partopsep=1ex,parsep=1ex,%
   label=(\Roman*)}%

\newlist{compactenumi}{enumerate}{5}%
\setlist[compactenumi]{topsep=0pt,itemsep=-1ex,partopsep=1ex,parsep=1ex,%
   label=(\roman*)}%

\newlist{compactitem}{itemize}{5}%
\setlist[compactitem]{topsep=0pt,itemsep=-1ex,partopsep=1ex,parsep=1ex,%
   label=\ensuremath{\bullet}}%

\numberwithin{figure}{section}%
\numberwithin{table}{section}%
\numberwithin{equation}{section}%

\newcommand{\ts}{\hspace{0.6pt}}

\newcommand{\dT}{\mathcalb{d}_T}%
\newcommand{\dTY}[2]{\dT\pth{#1,#2}}%
\newcommand{\Vor}{\mathcal{V}}
\newcommand{\VorX}[1]{\Vor\pth{#1}}%
\newcommand{\areaX}[1]{\mathrm{area}\pth{#1}}%
\newcommand{\Hd}{\mathcal{H}}%
\newcommand{\disk}{D}

\providecommand{\etal}{et~al.\xspace}
\renewcommand{\etal}{et~al.\xspace}
\usepackage{scalerel}
\newcommand{\tldO}{\scalerel*{\widetilde{O}}{j^2}}%

\newcommand{\seclab}[1]{\label{sec:#1}}
\newcommand{\secref}[1]{\HLink{Section}{sec:#1}}

\usepackage{qed_duck}

\bibliography{voronoi_prophet}

\begin{document}

\title{The Prophet and the Voronoi Diagram}

\author{Sariel Har-Peled\SarielThanks{Work on this paper was partially
      supported by NSF AF award CCF-2317241.  }}

\date{\today}

\maketitle

\begin{abstract}
    Consider a stream of $n$ random points (say, from the unit square) arriving one by one, where a player has to make an irreversible immediate decision for each arriving point whether to pick it. The player has to pick a single point, and the payoff is the area of the cell of the picked point, in the final Voronoi diagram of \emph{all} the points. We show that there is a simple strategy so that with probability $\geq 1 - \tldO(1/\sqrt{n})$, the player's payoff is only a constant factor smaller than the optimal choice (i.e., the one made by the prophet). This competitiveness is somewhat surprising, as this payoff is larger by a factor of $\Theta( \log n)$ than the average payoff.
\end{abstract}

\section{Introduction}

\paragraph*{The game.}

Imagine that a player is given a sequence of points
$P = \{p_1,\ldots, p_n\}$ one by one in the unit square $[0,1]^2$ (the
player knows the value of $n$ in advance). Whenever the player gets a
point, they can decide to use it as their desired site. The player can
choose only one such point during the game. The payoff is the area of the
Voronoi cell of the point picked in the final diagram formed by
\emph{all} the points of $P$. The key restriction is that the player has
to decide immediately upon seeing a point $p_i$ whether to pick it or
not---if the player decides not to pick a point, it can never be used.

A natural question is to devise a good strategy so that at the end of the
game, the expected payoff for the player is (up to a constant) the same
as the all-knowing (some would say all-annoying) prophet%
\footnote{Mark Twain: ``Prophecy is a good line of business, but it is
   full of risks.''}, who picks the site associated with the cell of
maximum area at the end of the game.

The setting is somewhat similar to the \emph{Secretary
   Problem}\footnote{Bad name---the secretary is rarely the problem.}
in optimal stopping theory. Imagine interviewing a sequence of candidates
for a position.  One interviews them sequentially (assigning each one a
score evaluating their abilities), say in random order. The interviewer
must make an irreversible, immediate decision whether to hire a candidate
once their interview is done. The target is to maximize the probability
of hiring (what seems to be) the best candidate. The classical strategy
is to interview a $\tfrac{1}{e}$ fraction of the candidates, and then
hire the first candidate that is better than any seen in the prefix. This
strategy has a probability of $\tfrac{1}{e}$ of hiring the best
candidate. There are numerous variants of this problem, with some of them
being useful for online auctions of ads in search results. For a nice
survey of the history of the problem, see Ferguson
\cite{f-wssp-89}\footnote{Of note, is Section 5 in \cite{f-wssp-89},
   describing how Kepler used a similar algorithm trying to decide who to
   merry as his second wife.}.  Coming up with such guarantees usually
involves proving a ``prophet inequality'' bounding the performance of the
suggested strategy versus the ``prophet''.

The difference between our game and the secretary problem is that there
is later interference in the quality of the choice the player
made---later points might shrink the chosen point's cell to be
(relatively) small. A better metaphor might be land speculation---a piece
of land bought early might be useless at the end of the game, if somebody
builds an obnoxious facility next to it \cite{cd-rofl-22}.  Thus, our
problem can be interpreted as modeling a land speculator's dilemma---should
they invest now, or wait for a better opportunity later on?

We are not aware of any work on prophet inequality with a model where the value of the final payoff might deteriorate after the decision is made.

\paragraph*{The Voronoi game.}
A related game, known as the Voronoi game, deals with the task of inserting points into a Voronoi diagram and trying to maximize the total area of the inserted points. For example, Cheong \etal \cite{chlm-orvg-04} showed that in an existing Voronoi diagram of the unit square of $n$ points, one can insert $n$ new points, and the new cells would have total area at least $1/2+c$, where $c >0$ is some absolute constant. This game is quite hard to analyze, but there are several known results, see \cite{accgo-cflvg-04}.

\paragraph*{The Voronoi game when a player has multiple choices.}

(This is not directly related to our work, but as the known results are
quite interesting, we discuss them here shortly.)  A game-theoretic
variant of the problem was studied by Boppana \etal
\cite{bhmm-vcg-16}. Here, each player is given a set of $m$ possible
locations, and each player can choose a location out of their available
locations, under the above payoff model. They show that even in 2d there
is no pure Nash equilibrium for $3$ players, where each player has two
possible locations (e.g., Rock-Paper-Scissors does not have a pure
equilibrium, so this situation is quite common).

As a reminder, a \emph{pure} equilibrium implies that each player has a unique choice of a point, such that in the resulting solution, no player can improve their payoff by changing their choice. The existence of a regular Nash equilibrium only guarantees that there is a stable mixed strategy (i.e., each player has a distribution over their choices) for each player.

Boppana \etal show that deciding if there exists a \emph{pure} Nash equilibrium is NP-hard for $m=4$ and one dimension (here, each player is assigned the Voronoi interval counterclockwise to its chosen location on the circle). Strangely, if the payoff is the length of the Voronoi cell (on the circle) for a player, then a pure equilibrium always exists.  To alleviate that, Boppana \etal study the game where the locations of the points provided for each player are chosen randomly, where they study the number of pure Nash equilibria (all in one dimension on the circle).

\paragraph*{Stochastic geometry.}

There is a vast amount of work on the geometric properties of Voronoi diagrams and similar geometric structures when the input points are picked randomly \cite{s-iig-53, s-iggp-04, obsc-stcav-00, cskm-sga-13}.

\paragraph{Our contribution (small as it might be).}

We introduce the game, as described above, and show that there is a
constant-factor competitive strategy. Informally, the idea is to wait
till roughly all but the last $O(\sqrt{n})$ points in the sequence had
been seen, and then pick the largest cell encountered in this final
suffix. The proof of correctness relies on showing that there are many
large cells, of similar size to the largest one, and one of them must pop
up in this suffix, and has low probability of being destroyed later on.

To this end, we show that many grid cells would have a single point in
them, despite each grid cell being ``large''. Specifically, the area of
each grid cell is $\Theta(\frac{\log n}{n})$.  A constant fraction of
these points are going to have large cells at the end of the game.  The
former follows from elementary computations, see \secref{balls_and_bins},
where we (re)prove some observations about balls into bins, and the
distribution of bins that have only a single point in them.

\secref{area} continues our bad habit of reproving known things,
providing several observations about the distribution and area of the
large and largest cells in a random Voronoi diagram.  We emphasize,
however, that none of these observations seems to be new.  Nevertheless,
for completeness, we provide self-contained proofs of them.  Our proofs
are elementary and (arguably) elegant, avoiding the heavier machinery
usually used in proving such things. The new proofs might be of
independent interest. In particular, we bound the size of the largest
cell in a random Voronoi diagram (\lemref{biggest_cell}), and prove that
many cells are almost as large (\lemref{many_large}). The latter is a
surprisingly easy implication of the (known) bounds on balls and bins
proved earlier.

Equipped with these tools, the player's strategy and analysis are provided in \secref{strategy}.

\paragraph*{Paper organization.}
In \secref{prelims}, we cover some preliminaries, including some
background on martingales and balls into bins.  In \secref{area}, we
prove the bounds needed on the distribution of the area of the largest
cells in a random Voronoi diagram.  In \secref{strategy}, we deploy these
tools to present and analyze the player's strategy that guarantees that
it is $O(1)$-competitive with probability close to one. Some concluding
remarks are given in \secref{conclusions}.

\section{Preliminaries}
\seclab{prelims}

All $\log$s in the paper are natural logs\footnote{All non-natural logs
   are considered to be an abomination, and would be actively avoided.}.

\paragraph*{Torus topology.}
In the following, consider the unit cube $\Hd_d = [0,1]^d$ under the torus topology.  In one dimension, the distance between two points $p_1,q_1 \in [0,1]$ is $\dTY{p_1}{q_1} = \min( |p_1 - q_1|, 1 - |p_1 - q_1|)$.  More generally, the \emphi{torus distance} between two points $p, q \in \Hd_d$ is
\begin{equation*}
    \dTY{p}{q} = \sqrt{ \sum_{i=1}^d \dTY{p_i}{q_i}^2 }.
\end{equation*}

\begin{remark}[Voronoi diagrams and torus topology]
    In the following, we deal with sets of points $P \subseteq [0,1]^2$, and
    their associated Voronoi diagram restricted to $[0,1]^2$, where one
    uses the torus distance as defined above (i.e., torus topology) to
    compute the diagram. To avoid annoying repetition, we emphasize that
    all Voronoi diagrams, in the following, are computed and considered
    under these settings.
\end{remark}

\subsection{A bit on balls into bins}
\seclab{balls_and_bins}

Imagine throwing $n$ balls into $m = \alpha n$ bins, for $\alpha = 4/\log n$, uniformly at random. For simplicity of exposition, assume that $n = t^4$, for some integer $t$. The probability that a specific bin has $k$ balls at the end of this process is
\begin{equation*}
    \tau_k = \binom{n}{k} p^k (1-p)^{n-k},
\end{equation*}
where $p = 1/m = 1/\alpha n$.

Let $X_k = f_k(b_1, \ldots, b_n)$ be the number of bins containing
exactly $k$ balls, if $b_1, \ldots, b_n$ are the locations of the $n$
balls thrown.  Observe that $\Ex{X_k} = m \tau_k$.  A standard way to
think about such processes is as an \emphw{exposure martingale}. That is,
let $Z_i = \ExCond{f_k}{b_1, \ldots b_i}$ (i.e., we expose the location
of the balls, one by one), for $i=0,\ldots, n$. Clearly, $Z_n = X_k$,
while $Z_0 = \Ex{X_k}$. It is not hard to verify that $f_k(\cdots)$ is
$2$-Lipschitz---moving a single ball from one bin to another changes the
status of at most two bins. That implies that $|Z_i - Z_{i-1}| \leq
2$. The sequence $Z_0, \ldots, Z_n$ is a martingale.

\begin{theorem}[Azuma's Inequality]
    \thmlab{Azuma}%
    Let $ Z_0, \ldots, Z_n$ be a martingale with $Z_0 =\mu$, and
    \begin{math}
        |Z_{i+1}-Z_i| \leq c_i,
    \end{math}
    for $i=0, \ldots, n-1$. Then, for any $\lambda > 0$, we have
    \begin{math}
        \Bigg.  \Prob{|Z_n- \mu| > \lambda \sqrt{n} \ts} <
        2\exp\pth{\Bigl.-\lambda^2/(2\sum_i c_i^2)}.
    \end{math}
\end{theorem}

\subsubsection{Proving concentration of empty and singleton bins}

The following are pretty standard estimates---we include the detailed
proofs for the sake of completeness\footnote{This was also added to
   ensure this writeup is not too short. A reader finding this manuscript
   still too short should read between the lines.}.

\begin{lemma}
    Consider throwing $n$ balls into $m = \alpha n$ bins, for $\alpha = \beta/\log n$, where $\beta \geq 4$ is a constant and $n \geq c_0$, where $c_0$ is some \emph{sufficiently} large constant.  Let $X_0$ be the number of empty bins. We have $X_0 = \Theta( n^{1-1/\beta}/ \log n)$ with probability $\geq 1 - n^{-c}$, where $c \geq 8$ is a fixed arbitrary large constant.
\end{lemma}
\begin{proof}
    For $k=0$, using $1-x \leq e^{-x} $, we have
    \begin{align*}
        \Ex{X_0}
        &=
        \alpha n \pth{1 -\frac{1}{\alpha n}}^n
        \leq
        \alpha n \exp \pth{ -\frac{1}{\alpha n}}^n
        =
        \alpha n \exp \pth{ -\frac{1}{\alpha }}
        \\&%
        =%
        \alpha n \exp \pth{ -\frac{\log n}{\beta }}
        =
        \beta \frac{n^{1-1/\beta}}{\log n}.
    \end{align*}
    Using $1-x \geq \exp(-x -x^2)$, we have
    \begin{equation*}
        \Ex{X_0}
        \geq
        \alpha n \exp \pth{ -\frac{1}{\alpha n} - \frac{1}{\alpha^2 n^2}}^n
        =
        \alpha n \exp \pth{ -\frac{1}{\alpha }} \exp\pth{ -\frac{1}{\alpha^2 n}}
        \geq
        \frac{\beta n^{1-1/\beta}}{2\log n},
    \end{equation*}
    for $n$ sufficiently large.  As the value of $X_0$ forms an exposure martingale, by \thmrefY{Azuma}{Azuma}'s inequality, we have that $\Prob{|X_0 - \Ex{X_0}| \geq \lambda \sqrt{n}} \leq \exp( -\lambda^2 /8)$. The claim now follows, by taking (say) $\lambda = n^{1/8}$.
\end{proof}

\begin{remark}
    Somewhat confusingly, in the above lemma, we throw $n$ balls into
    significantly fewer bins (i.e., fewer by a logarithmic factor). It is
    easy to verify that there are no empty bins, with high probability,
    if the number of bins is (say) $\beta \tfrac{n}{\log n}$ if
    $\beta=1/10$. The lemma states that if $\beta \geq 4$, then the number of empty
    bins is polynomially large (and with high probability).  So the exact
    value of $\beta$ matters, as there is a phase transition in the interval
    $\beta \in [1/10, 4]$.
\end{remark}

\begin{lemma}
    \lemlab{one_point}%
    For the number $X_1$ of bins containing exactly one ball, we have $X_1 \geq n^{1-1/\beta}/ (4 \log n)$ with high probability (i.e., $X_1 =\Theta( n^{1-1/\beta}/ \log n)$ with high probability).
\end{lemma}
\begin{proof}
    Arguing as above, as $\alpha = \beta/\log n$, we have
    \begin{align*}
        \Ex{X_1}
        &=
        \alpha n \cdot n \frac{1}{\alpha n}
        \Bigl(1 -\frac{1}{\alpha n} \Bigr)^{n-1}
        \leq
        2n \exp \Bigl( -\frac{1}{\alpha n}\Bigr)^{n}
        =
        2n \exp \Bigl( -\frac{\log n}{\beta} \Bigr)
        \\&
        =
        2n^{1-1/\beta}.
    \end{align*}
    Similarly, we have
    \begin{math}
        \Ex{X_1} = n \pth{1 -\frac{1}{\alpha n}}^{n-1} \geq n \exp \pth{ -\frac{1}{\alpha } - \frac{1}{\alpha^2 n }} \geq \frac{n^{1-1/\beta}}{2}.
    \end{math}
    Observing that the value of $X_1$ can be evaluated as an exposure martingale implies the result by Azuma's inequality.
\end{proof}

\section{Distribution cell areas in a Voronoi diagram}
\seclab{area}

\subsection{The largest Voronoi cell is not that large}

\begin{lemma}
    \lemlab{biggest_cell}%
    Consider throwing $n$ points $P = \{ p_1, \ldots, p_n \}$, uniformly and independently, into $\Hd = [0,1]^2$, and let $\Vor = \VorX{P}$ be the Voronoi diagram of the points of $P$. Let $A$ be the area of the largest Voronoi cell in $\Vor$. Then, for any constant $c >1$, we have
    \begin{math}
        \Prob{ \smash{A \geq 4c \frac{\log n}{n}} } \leq \frac{1}{n^{2c-4}}.
    \end{math}
\end{lemma}
\begin{proof}
    Let $p$ be a point in $P$, and assume it has a cell $C$ in $\Vor$ of area $\gamma$. Let $u$ be the furthest point from $p$ in $C$, with $r = \dTY{p}{u}$. Observe that $\areaX{C} \leq \pi r^2$, implying that $r = \dTY{p}{u} \geq \sqrt{\areaX{C}/\pi}$.

    Let $G$ be a set of $n^{4}$ points formed by a uniform grid in
    $\Hd$, such that the distance of any point in $\Hd$ to the
    nearest point in $G$ is at most $\sqrt{2}/n^2$. Observe that if there
    is a disk of radius $r$, for $r > 1/n$, in $\Hd$ containing no points
    of $P$, then there is such a disk of radius
    $r - \sqrt{2}/n^2 \geq R = r(1-1/8) $ centered at a point of $G$.

    If the area of $C$ is at least $A = 4c \frac{\log n}{n}$, then $r \geq \sqrt{A/ \pi} \geq \sqrt{c \tfrac{\log n}{n}}$, and $R \geq \tfrac{7}{8} \sqrt{c \tfrac{\log n}{n}}$.  By the union bound, the probability of having a disk of radius $R$ centered at a point of $G$ that is empty is at most
    \begin{equation*}
        n^4 \pth{ 1 - \pi R^2}^n
        \leq
        n^4 \exp\pth{ -\pi R^2 n }
        \leq
        n^4 \exp\pth{ -\pi \frac{49 c \log n} {64n} \cdot n }
        \leq
        n^{4-2c}.
    \end{equation*}
\end{proof}

\subsection{There are many large Voronoi cells}

\subsubsection{A Voronoi cell of a point vs the boundary of a square}

In the following, let $C$ be a square in the plane. For a point $p \in C$, let $D_p$ be the locus of points closer to $p$ than any point on the boundary of $C$. Consider the function $f(p ) = \areaX{D_p} / \areaX{C}$.

\begin{lemma}
    \lemlab{center}%
    For a random point $p \in C$, we have $\Prob{ f(p) \geq 1 / 15} \geq 1/15$.
\end{lemma}
\begin{proof}
    Let $\ell$ be the side length of $C$, and let $C'$ be the square of
    length $\delta = \ell / ( 1 + 2 \sqrt{2})$ centered at the center of
    $C$, and assume that both $C$ and $C'$ are axis-aligned. We require
    that the minimum distance of any point in $C'$ to $\partial C$ is at
    least $\sqrt{2}\delta$. Namely, we have that
    $\ell = 2 \sqrt{2} \delta + \delta$, see \figref{p:in:square}---this
    is how the value of $\delta$ was chosen.
    \begin{figure}[h]
        \centerline{\includegraphics{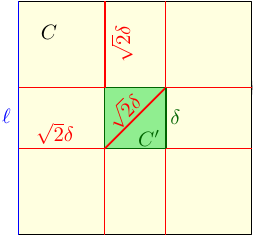}}
        \caption{A point inside $C'$ contains it in its Voronoi cell.}
        \figlab{p:in:square}%
    \end{figure}%
    Thus, the probability that a random point $p$ falls into $C'$ is $1 / (1 + 2\sqrt{2})^2 \geq 1/15$. If $p$ falls inside $C'$, then $C' \subseteq D_p$ (i.e., all the points of $C'$ are closer to $p$ than any point of $\partial C$), implying the claim.
\end{proof}

\begin{remark}
    A significantly more tedious argument shows that
    \begin{equation*}
        \Prob{ f(p) \geq \areaX{C} / 12} \geq 1/12.
    \end{equation*}
    Numerical simulations suggest that the correct statement of this type
    is $\Prob{ f(p) \geq \areaX{C} / 6} \geq 1/6$. A simple and elegant proof
    of this would be quite interesting, although for our purposes, any
    constant suffices. The insistence of having the same fraction on both
    sides of the inequality is only for vanity---any constant should
    work.
\end{remark}

\subsubsection{Many large cells}

\begin{lemma}
    \lemlab{many_large}%
    Let $P$ be a set of $n$ points picked uniformly and independently at
    random from $[0,1]^2$, where $n$ is sufficiently large. Then, there
    is a constant $\gamma$ such that in the Voronoi diagram of $P$ in
    $[0,1]^2$, there are at least $\sqrt{n}/30$ cells of area at least
    $\gamma (\log n) / n$. Furthermore, these cells are relatively
    fat---there is a radius $R = \Theta(\sqrt{(\log n)/n }) $, such that
    the disk of radius $R$ centered at the site of these large cells
    contains no other point of $P$.
\end{lemma}

\begin{proof}
    Let $t$ be the largest integer such that
    \begin{equation*}
        4\frac{n}{\log n}
        \leq
        t^2
        \leq
        8\frac{n}{\log n}.
    \end{equation*}
    We partition $[0,1]^2$ into $m=t^2$ cells using a uniform grid
    $t \times t$. Observe that $m = \beta n / \log n $, for some
    $4 \leq \beta \leq 8$. Thus, by \lemref{one_point}, there are at least
    $X_1 \geq n^{1-1/\beta}/ (4 \log n) \geq \sqrt{n}$ cells of this grid
    that contain exactly one point, and this happens with high
    probability.

    Consider such a cell $C$ with a single point $p$ in it (a
    ``singleton'' cell). By \lemref{center}, with probability
    $\geq \psi=1/15 $, the point $p$ is closer to at least
    $\psi$-fraction of this cell, than to the boundary of the cell. Thus,
    the Voronoi cell of $p$ in the Voronoi diagram of $P$, with
    probability at least $\psi$, is at least of area
    \begin{equation*}
        \xi = \psi / t^2 \geq \frac{\log n}{15 \cdot 8 n} = \frac{\log n}{120 n}.
    \end{equation*}
    Namely, the cell of this point is quite large with probability at
    least $\psi$. The probability that all such singleton cells fail to
    generate a large Voronoi cell is at most
    $(1-\psi)^{X_1} = O\bigl( \exp( -\sqrt{n }/15 )\bigr)$, which is
    absurdly small. Moreover, let $Y$ be the number of singleton cells in
    the grid that succeed in generating a cell of size $\geq \xi$. We have
    that $X_1 \geq \sqrt{n}$, and thus
    $\mu = \Ex{Y} \geq \psi X_1 \geq \sqrt{n} /15$. A standard application of
    Chernoff's inequality now implies that
    $\Prob{Y \leq \sqrt{n}/30} = \Prob{ Y \leq \mu /2 } \leq \exp( -\mu /8 ) \leq
    \exp( - \sqrt{n}/ 128 )$, which is still absurdly small. Namely, with
    high probability, there are $\Omega( \sqrt{n} )$ Voronoi cells of
    area $\geq \xi$.

    The claim about the radius is a byproduct of the above
    analysis---each successful singleton cell has a point close to the
    cell's center, and a disk of radius proportional to the sidelength of
    the cell is fully contained in the Voronoi cell of $p$, see
    \figref{p:in:square}.
\end{proof}

\section{The player's strategy}
\seclab{strategy}

In the following, we assume $n$ is sufficiently large, and the player is exposed to the random points of $P = \{ p_1, \ldots, p_n\}$ one by one.  As a reminder, the points are picked uniformly at random from the unit square $[0,1]^2$.  For any $i$, let $P_i =\{p_1, \ldots, p_i\}$ denote the prefix of the first $i$ points. The player's strategy is to ignore the first (say) $ n - f$ points, where
\begin{equation}
    f = c \sqrt{n} \log n,
    \eqlab{f:value}
\end{equation}
and $c$ is a sufficiently large constant to be determined shortly. The player then picks the first point $p_i$, in the remaining suffix of the permutation, such that the disk $\disk_i$ of radius $R = \Theta( \sqrt{(\log n)/n} ) $ (the exact value is specified by the analysis of \lemref{many_large}) centered at $p_i$, contains no other point of $P_i$. The player made its decision, and now has to hope for the best.

\begin{theorem}
    Let $P$ be a sequence of $n$ points picked uniformly and randomly from $\Hd = [0,1]^2$. The sequence is provided in an online fashion, one point at a time, and the player has to irrevocably pick one of the points in the sequence when it is provided (without knowing the future points). Then, the above strategy results in the player having a Voronoi cell of size $\Omega( \nabla ) $, where $\nabla = \tfrac{\log n}{ n}$ and the largest cell in the Voronoi diagram of $P$ (under the torus topology) has area at most $O( \nabla )$. That is, the player's strategy is $O(1)$-competitive with the optimal strategy that has full knowledge of $P$. This guarantee holds with probability $\geq 1 - O\pth{ \frac{\log^2 n}{\sqrt{n}} }$.
\end{theorem}

\begin{proof}
    By \lemref{many_large}, there are at least $U \geq c' \sqrt{n}$ sites in $P$, containing a disk of radius $R = \Theta( \sqrt{(\log n)/n} )$ centered at the site, that does not contain any other point of $P$, where $c'$ is some absolute constant. This implies that in expectation, by \Eqref{f:value}, there are at least
    \begin{equation*}
        N
        =
        \frac{U}{n}f
        \geq
        \frac{c' \sqrt{n }}{n}
        \cdot c \sqrt{n} \log n
        \geq
        c'' \log n,
    \end{equation*}
    such sites in the suffix $p_{n-f+1}, \ldots, p_n$, and one can make the leading constant $c'' $ arbitrarily large, by increasing the constant $c$ in the strategy used to define $f$. The probability that none of these $U$ large sites fall in the suffix $P \setminus P_{n-f}$ is clearly polynomially small. For those of little faith, this probability is
    \begin{align*}
        \eta
        &=
        \frac{n-U}{n}\cdot \frac{n-U-1}{n-1} \cdots \frac{n-U-f+1}{n-f+1}
        \leq
        \pth{\frac{n-U}{n}}^f
        \\&
        \leq
        \exp \pth{ - \frac{U f}{n}}
        =
        \exp \pth{ - N}
        \\&
        =
        \frac{1}{n^{c''}},
    \end{align*}
    as one imagines picking the last $f$ points first, avoiding the $U$ ``forbidden'' large sites. As we can make this bad probability arbitrarily small, this means that with high probability, a large cell would be encountered in the last $f$ sites.

    All the player can do is to check that the disk $D_i$ is empty of points of $P_{i-1}$. By the above, the player would encounter such a disk in the last $f$ sites, with high probability. Sadly, if somewhat unlikely, one of the final ($n-i \leq f$) remaining sites might still hit $D_i$. Being pessimistic, this would be considered a failure. The probability for that, using the union bound, is
    \begin{equation*}
        \areaX{D_i} f
        =%
        O\pth{ \frac{\log n}{n} \sqrt{n } \log n }
        =%
        O\pth{ \frac{\log^2 n}{\sqrt{n}} }.
    \end{equation*}

    If the disk $D_i$ the player chose is not hit, then at least a
    quarter of its area belongs to the final Voronoi cell of
    $p_i$. Indeed, this cell is at least the Voronoi cell of $p_i$ vs all
    the boundary points of $D_i$.
\end{proof}

\begin{remark}
    One can improve the probability of failure by observing that two sites have to hit $D_i$ for the area the player gets to be catastrophically smaller. Indeed, a single site hitting $D_i$ can ``steal'' at most half the area of the Voronoi cell assigned to $p_i$ before it got hit.
\end{remark}

\section{Conclusions}
\seclab{conclusions}

We propose a new geometric variant of the secretary problem where the
payoff is a function computed at the end of the game---specifically,
the area of the Voronoi cell the player picked. Unlike the secretary
problem, the payoff is not guaranteed at the time of selection and is
subject to deterioration due to interference by later elements appearing
in the sequence. We are unaware of previous work on such variants, and we
think they are both interesting and natural.

Our result should extend naturally to $d >2$ in a straightforward fashion, as we used surprisingly few geometric properties of Voronoi diagrams. Our analysis relied on showing that many sites are far from other sites and thus have large cells, enabling the player to wait for almost the end of the sequence before making their choice.

There are many natural problems for further research---different interference models, what if the sites are predetermined, but their order is random. What if the player can make $k$ choices (and keep all of them, or only the best one), etc.

\paragraph*{Acknowledgments.}
The author thanks Farouk Harb for his useful comments on the manuscript.

\printbibliography

\end{document}